\documentclass[11pt]{article}

\usepackage{amsthm}
\theoremstyle{plain}

\newtheorem*{theorem*}{Theorem}
\newtheorem{remark}{Remark}
\newtheorem{theorem}{Theorem}[section]
\newtheorem*{claim*}{Claim}
\newtheorem{claim}[theorem]{Claim}
\newtheorem*{lemma*}{Lemma}
\newtheorem{lemma}[theorem]{Lemma}
\newtheorem{definition}{Definition}[section]
\newtheorem{problem}[theorem]{Problem}
\newtheorem{hypothesis}[theorem]{Hypothesis}
\newtheorem*{proof*}{proof}

\newcommand\addtag{\refstepcounter{equation}\tag{\theequation}}

\usepackage{color}
\usepackage[colorlinks,linkcolor=blue,citecolor=blue,urlcolor=blue]{hyperref}
\usepackage[capitalize]{cleveref}
\usepackage{color}
\usepackage[colorlinks,linkcolor=blue,citecolor=blue,urlcolor=blue]{hyperref}
\usepackage[capitalize]{cleveref}
\usepackage{amsmath,amssymb}
\usepackage{graphicx}
\usepackage{caption}
\usepackage{color}
\usepackage{dcolumn}
\usepackage{bm}

\newcommand{\SSE}{{\sc Small-Set Expansion}}
\newcommand{\MLA}{{\sc Minimum Linear Arrangement}}
\newcommand{\kBP}{{\sc $k$-Balanced Partitioning}}
\newcommand{\SC}{{\sc Sparsest Cut}}

\definecolor{background-color}{gray}{0.98}

\begin{document}

\title{Approximate Hierarchical Clustering via \\ Sparsest Cut and Spreading Metrics}
\author{Moses Charikar\thanks{Computer Science Department, Stanford University, moses@cs.stanford.edu, vaggos@stanford.edu} \ \ \ \ \ \ \ \ \ \ \ \ \ \ \ \ \ \ Vaggos Chatziafratis$^*$}

\maketitle

\begin{abstract}

Dasgupta recently introduced a cost function for the hierarchical clustering
of a set of points given pairwise similarities between them.
He showed that this function is $NP$-hard to optimize, but a top-down recursive partitioning heuristic based on an 
$\alpha_n$-approximation algorithm for uniform sparsest cut gives an approximation of 
$O(\alpha_n\log n)$ (the current best algorithm has $\alpha_n=O(\sqrt{\log n})$). We show that the aforementioned sparsest cut heuristic in fact obtains an $O(\alpha_n)$-approximation.
The algorithm also applies to a generalized cost function studied by Dasgupta.
Moreover, we obtain a strong inapproximability result, showing that the Hierarchical Clustering objective is hard to approximate to 
within any constant factor assuming the \textit{Small-Set Expansion (SSE) Hypothesis}.
Finally, we discuss approximation algorithms based on convex relaxations.
We present a spreading metric SDP relaxation for the problem and show that it has integrality gap at most $O(\sqrt{\log n})$. 
The advantage of the SDP relative to the sparsest cut heuristic is that it provides an explicit lower bound on the optimal solution and could potentially yield an even better approximation for hierarchical clustering.
In fact our analysis of this SDP served as the inspiration for our improved analysis of the sparsest cut heuristic.
We also show that a spreading metric LP relaxation gives an $O(\log n)$-approximation.

\end{abstract}

\clearpage
\section{Introduction}\label{sec:introduction}

Hierarchical Clustering (HC) of a data set is a recursive partitioning of the data into clusters.
Such methods are widely used in data analysis.
To be more formal, in the Hierarchical Clustering (HC) problem, the input is a weighted undirected graph $G = (V, E, w)$. Each data point corresponds to a node in the graph and edges connect similar points. The heavier the edge weight the stronger the similarity between the data points. 
The goal is to produce a partitioning of the data into successively smaller clusters, starting from the original graph $G$ as the initial cluster and ending with $n$ singleton clusters. The HC is represented as a tree with leaves corresponding to data points and internal nodes corresponding to clusters in the hierarchy. 

Such a hierarchical decomposition of data has several advantages over \textit{flat clustering} ($k$-means, $k$-center etc): firstly, there is no need to fix the number $k$ of clusters we want to create; secondly, large datasets are understood simultaneously at many levels of granularity and thirdly, many greedy heuristics with provable approximation guarantees (\cite{dasguptaSTOC16}) can be used to construct it. 

Despite its important applications for many scientific areas such as biology (e.g. gene expression), data analysis, phylogenetics, social sciences and statistics, HC and the algorithms we use to solve it in practice are not yet well-understood. Many heuristics have been proposed, some of which are based on a natural ``bottom-up'' approach by recursively merging data that are similar: at the beginning each data point is a separate cluster and we start merging them based on their similarity as we go up the hierarchy. These are the so-called agglomerative methods that are provided by standard Data Analysis packages and include for example single-linkage, average linkage etc. (\cite{jain1988algorithms,jain1999data,balcan2014robust,hastie2009}). These methods are specified procedurally rather than in terms of an objective function for HC; this lack of objective functions for the problem of HC was addressed by the recent work of Dasgupta (\cite{dasguptaSTOC16}). He introduced a simple cost function that, given pairwise
similarities between data points, assigns a score to any possible tree on those points. The tree corresponds to the hierarchical decomposition of the data and its score reflects the quality of the solution.

Let $T$ be any rooted (not necessarily binary) tree that has a leaf for each point in our dataset.
For a node $u$ in $T$, we denote with $T[u]$ the subtree rooted at $u$, and with $\bold{leaves}(T[u]) \subseteq V$ we denote the leaves of this
subtree. For leaves $i, j \in V$ , the expression $i \lor j$ denotes their lowest common ancestor in $T$, i.e. $T[i \lor j]$ is the smallest subtree whose leaves include both $i$ and $j$. 
The following cost function is the HC cost function:
\begin{align}
cost_G(T)=\sum_{ij \in E} w_{ij} |\bold{leaves}(T[i\lor j])|.
\label{obj:clustering}
\end{align}

We observe that a heavy edge should not be cut at the top of the tree because it would cause a high cost due to the term $|\bold{leaves}(T[i\lor j])|$ that would be large.
For example, if an edge $\{i, j\}$ of unit weight is cut at the first split of the data, then we pay $n$. If it is cut further down, in a subtree that contains a $\delta$ fraction of the data, then
we pay $\delta n$. We would like to find a tree $T^*$ that minimizes the above cost. It is not difficult to see that there must always exist an optimal tree that is binary, since by converting any split that creates more than two subtrees to a sequence of binary splits, we can never increase the cost. A generalized version for the cost function is also considered in \cite{dasguptaSTOC16}:
\begin{align}
cost_G(T)=\sum_{ij \in E} w_{ij} f(|\bold{leaves}(T[i\lor j])|).
\label{obj:gen_clustering}
\end{align}

\subsection{Related Work}
Dasgupta introduced the cost function (\ref*{obj:clustering}) and explained why it is a good objective function for hierarchical clustering. He presented some interesting special cases (e.g. planted partitions) for which optimizing (\ref*{obj:clustering}) actually finds the correct underlying HC. 
He showed that optimizing it is an NP-hard problem and showed that a simple heuristic based on an $\alpha_n$ approximation for Sparsest Cut will achieve a factor $O(\alpha_n \cdot \log n)$ approximation. The current best $\alpha_n$ ratio for Sparsest Cut is $O(\sqrt{\log n})$ from a breakthrough result of \cite{arora2009expander}. The heuristic starts by taking Sparsest Cut for the input graph $G$, splitting it into $(G_1,G_2)$ and then applying Sparsest Cut recursively to the pieces $G_1,G_2$. Dasgupta also proved that a slightly modified heuristic yields basically the same approximation guarantee for optimizing (\ref*{obj:gen_clustering}).

Another natural approach for dealing with HC is to try to optimize standard popular cost functions for \textit{flat clustering}, such as the $k$-means, $k$-median or $k$-center(\cite{dasgupta2005performance,plaxton2003approximation,lin2010general}). However, with this approach, it is necessary to cut at some level the hierarchy so that we get $k$ clusters at the end or use many different values of $k$ to achieve a satisfying depth of decomposition. 

People have also studied methods of HC in terms of statistical consistency (\cite{Hartigan1985,chaudhuri2014consistent, eldridge2015beyond}), where data points are sampled from a fixed underlying distribution and we are interested in the convergence of the tree structure obtained from the data as the sample size goes to infinity. Only a few methods are known to be consistent (\cite{chaudhuri2014consistent, eldridge2015beyond}). Furthermore, the authors of \cite{balcan2014robust} study the performance of agglomerative clustering techniques in the
presence of noise and they propose a new algorithm that is more robust and performs better in cases with noisy data where traditional agglomerative algorithms fail.

Recently we were informed that independently of our work, Roy and Pokutta~\cite{pokuttaNIPS16} got a similar result for Hierarchical Clustering via spreading metrics. In particular, they used an LP relaxation based on ultrametrics to prove an $O(\log n)$ approximation. The LP relaxation they formulated was similar to ours but we viewed it as a vector programming relaxation and we managed to obtain an $O(\sqrt{\log n})$ approximation.  As far as their analysis is concerned they used the extensively studied (in the context of graph partitioning) idea of sphere growing (see~\cite{leighton1988approximate, garg1993approximate,even1999fast, charikar2003clustering, even2000divide}). On the other hand, we got our initial $O(\log n)$ approximation by proving that the hierarchical clustering objective function falls into the \textit{divide and conquer approximation algorithms via spreading metrics} paradigm of~\cite{even2000divide} and combining it with a result of Bartal~(\cite{bartal2004graph}). Finally, they also gave the same constant factor inapproximability result as we did, based on the small set expansion hypothesis.


\subsection{Our results and structure of the paper}
We 
show (\hyperref[sec:better-SpCut]{Section~\ref*{sec:better-SpCut}}) that the recursive sparsest cut (RSC) algorithm
that uses any $\alpha_n$-approximation algorithm for uniform sparsest cut achieves an $O(\alpha_n)$ approximation for hierarchical clustering, shaving a $\log n$ factor from Dasgupta's analysis. The analysis can be modified to prove that the same guarantee holds even for the generalized cost function (\ref*{obj:gen_clustering}). We also present (\hyperref[sec:hardness]{Section~\ref*{sec:hardness}}) a strong inapproximability result for HC, in particular, that it is hard to approximate HC to within any constant factor assuming the Small Set Expansion (SSE) Hypothesis. In~\hyperref[sec:sdp-approx]{Section~\ref*{sec:sdp-approx}}, we present an SDP relaxation based on spreading metrics with integrality gap at most $O(\sqrt{\log n})$ for HC. The advantage of the SDP relative to the sparsest cut heuristic is that it provides an explicit lower bound on the optimal solution and could potentially yield an even better approximation for hierarchical clustering. In fact, we first developed a rounding algorithm for this SDP and our analysis later served as the inspiration for our improved analysis of the sparsest cut heuristic for both cost functions (\ref*{obj:clustering}) and (\ref*{obj:gen_clustering}). Finally, we show how the spreading metrics paradigm of~\cite{even2000divide} in combination with a result of Bartal~\cite{bartal2004graph} (\hyperref[sec:spreading]{Appendix~\ref*{sec:spreading}}) can be exploited in order to get an $O(\log n)$ approximation for hierarchical clustering via a linear program (\hyperref[sec:lp-approx]{Appendix~\ref*{sec:lp-approx}}). We conclude in \hyperref[sec:conclusion]{Section~\ref*{sec:conclusion}} with questions for further research. Some preliminaries are deferred to the \hyperref[sec:prelim]{Appendix \ref{sec:prelim}} and omitted proofs are given in \hyperref[sec:omitted proofs]{Appendix \ref{sec:omitted proofs}}.

A key idea behind our analysis of the recursive sparsest cut algorithm as well as the formulation of the SDP relaxation is to view a hierarchical clustering of $n$ data points as
a collection of partitions of the data, one for each level $t = n-1,\ldots,1$.
Here the partition for a particular level $t$ consists of maximal clusters in the hierarchical clustering of size at most $t$.
When we partition a cluster of size $r$, we charge this to levels $t \in [r/4,r/2]$ of this collection of partitions.
This is crucial for eliminating the $\log n$ term in the approximation guarantee.




\section{Better Analysis for Recursive Sparsest Cut (RSC)}\label{sec:better-SpCut}

As discussed previously, 
Dasgupta~\cite{dasguptaSTOC16} showed that
a simple top-down Recursive Sparsest Cuts (RSC) heuristic 
that uses an $\alpha_n$-approximation algorithm for uniform sparsest cut gives an approximation of $O(\alpha_n \log n)$ for hierarchical clustering. 
More precisely, the RSC heuristic starts from the given graph $G=(V,E)$, uses any $\alpha_n$-approximation algorithm for sparsest cut, thus splitting $G$ into $(G_1,G_2)$ and then recurses on $G_1$ and $G_2$. The output is a binary tree of the sequence of cuts performed by the algorithm.

In this section, by drawing inspiration from our SDP construction and analysis presented later in \hyperref[sec:sdp-approx]{Section~\ref*{sec:sdp-approx}}, we present an improved analysis for this simple heuristic, dropping the $\log n$ factor and showing that it actually yields an $O(\alpha_n)$ approximation. This is satisfying since any improvement for Sparsest Cut would immediately yield a better approximation result for hierarchical clustering. Additionally, fast algorithms (i.e. nearly linear time algorithms) for Sparsest Cut (\cite{sherman2009breaking}) render the heuristic useful in practice. 
%

\subsection{Analysis of RSC heuristic}
Let the given graph be $G=(V,E)$. We suppose for clarity of presentation that it is unweighted; the analysis applies directly to weighted graphs and later, we see how to generalize it for more general cost functions. Let OPT be the optimal solution for hierarchical clustering (we abuse notation slightly by using OPT to denote both the solution as well as its objective function value).
Let $OPT(t)$ be the maximal clusters in OPT of size at most $t$. Note that $OPT(t)$ is a partition of $V$.
We denote $E_{OPT}(t)$ the edges that are cut in $OPT(t)$, i.e. edges with end points in different clusters in $OPT(t)$. For convenience, we also define $E_{OPT}(0) \triangleq E$.
%
%
\begin{claim}
$OPT=\sum_{t=0}^{n-1} |E_{OPT}(t)|$.
\end{claim}
\begin{proof}
%
%
Consider any edge $(u,v) \in E$.
Suppose that the size of the minimal cluster in OPT that contains both $u$ and $v$ is $r$. Then the contribution of $(u,v)$ to the LHS is $r$. On the other hand, $(u,v) \in E_{OPT}(t)$ for all $t \in \{0,\ldots,r-1\}$. Hence the contribution to the RHS is also $r$.
\end{proof}
It will be convenient to use the following bound that is directly implied by the above claim:
\[2OPT = 2\cdot \sum_{t=0}^{n-1} |E_{OPT}(t)| \ge \sum_{t=0}^{n} E_{OPT}( \lfloor t/2\rfloor) \addtag \label{eq:one} \]

Let's look at a cluster $A$ with size $|A|=r$ in the solution produced by RSC. Using a sparsest cut approximation algorithm, we create two clusters $B_1,B_2$ with sizes $s, (r-s)$ respectively, with $B_1$ being the smaller, i.e. $s\le \lfloor r/2 \rfloor$. The contribution of this cut to the hierarchical clustering objective function is: $ |E(B_1,B_2)| \cdot r$.
We basically want to charge this cost to $OPT(\lfloor r/2 \rfloor)$ and for that we first observe that the edges cut in $OPT(\lfloor r/2 \rfloor)$, when restricted to the cluster $A$ (i.e. having both endpoints in $A$), satisfy the following:
\[ s \cdot |E_{OPT}(\lfloor r/2 \rfloor) \cap A| \le \sum_{t=r-s+1}^r |E_{OPT}(\lfloor t/2 \rfloor) \cap A|. \addtag \label{eq:two}\] 	

This follows easily from the fact that 
$|E_{OPT}(t) \cap A| \leq  |E_{OPT}(t-1) \cap A|$.
Now in order to explain our charging scheme, let's look at the partition $A_1,..., A_k$ induced inside the cluster $A$ by $OPT(\lfloor r/2 \rfloor) \cap A$, where by design the size of each $|A_i| = \gamma_i |A|$, $\gamma_i \le 1/2$. We have:
\begin{align*}
\dfrac{|E(A_i, A\setminus A_i)|}{|A_i| |A\setminus A_i|} = \dfrac{|E(A_i, A\setminus A_i)|}{\gamma_i (1-\gamma_i)r^2}, \forall i \in {1,..., k} 
\end{align*}     

We take the minimum over all $i$ (an upper bound on the sparsest cut in $A$) and we have:
\begin{align*}
\min_i \dfrac{|E(A_i, A\setminus A_i)|}{\gamma_i (1-\gamma_i)r^2} \le \dfrac{\sum_i |E(A_i, A\setminus A_i)|}{\sum_i \gamma_i (1-\gamma_i)r^2} \le 2\cdot \dfrac{|E_{OPT}(\lfloor r/2 \rfloor)\cap A|}{r^2/2} = 4 \cdot \dfrac{|E_{OPT}(\lfloor r/2 \rfloor)\cap A|}{r^2}
\end{align*}     

The first inequality above, trivially follows by definition for the minimum and the second inequality holds because $\sum_{i=1}^k \gamma_i = 1, \sum_{i=1}^k \gamma_i^2 \le 1/2$ and the factor of 2 is introduced since we double counted every edge. We partition $A$ using an $\alpha_n$-approximation for sparsest cut and so:
\begin{align*}
\dfrac{|E(B_1,B_2)|}{s(r-s)} \le \alpha_n \cdot \dfrac{4}{r^2} \cdot |E_{OPT}(\lfloor r/2 \rfloor) \cap A|
\end{align*}     
since the RHS (without the $\alpha_n$ factor) is an upper bound of the optimal sparsest cut value.
The contribution of this step to the hierarchical clustering objective function is:
\begin{align} \addtag \label{eq:three}
r |E(B_1, B_2)| \le \dfrac{4\alpha_n s (r-s)}{r} \cdot |E_{OPT}(\lfloor r/2 \rfloor) \cap A| \le 4\alpha_n s \cdot |E_{OPT}(\lfloor r/2 \rfloor) \cap A|.
\end{align}

We claim the following:

\begin{claim}
Let $A$ be a cluster of size $r_A$ in our hierarchical clustering solution, that we split into 2 pieces $(B_1, B_2)$ of size $s_A, r_A -s_A$ respectively with $|B_1| \le |B_2|$. Then, summing over all clusters $A$: \\
\[\sum_{A} \ \sum_{t=r_A -s_A +1}^{r_A} |E_{OPT}(\lfloor t/2 \rfloor) \cap A| \le \sum_{t=1}^n |E_{OPT}(\lfloor t/2 \rfloor)| \addtag \label{eq:four} \]
\end{claim}

\begin{proof}
For a fixed value of $t$ and $A$, the LHS is: $|E_{OPT}(\lfloor t/2 \rfloor)\cap A|$. Consider which clusters $A$ contribute such a term to the LHS. From the fact that $r_A-s_A+1\le t\le r_A$, we need to have that $|B_2| < t$ and since $B_2$ is the larger piece that was created when $A$ was split, we deduce that $A$ is a \textbf{minimal} cluster of size $|A| \ge t > |B_2|\ge |B_1|$, i.e. if both $A$'s children are of size less than 
%
%
$t$,
then this cluster $A$ contributes such a term. The set of all such $A$ form a disjoint partition of  
%
$V$
because of the definition for minimality (in order for them to overlap in the hierarchical clustering, one of them needs to be ancestor of the other and this cannot happen because of minimality). Since $E_{OPT}(\lfloor t/2 \rfloor)\cap A$ for all such $A$ forms a disjoint partition of $E_{OPT}(\lfloor t/2 \rfloor)$,  the claim follows by summing up over all $t$.
\end{proof}

\begin{theorem}
Given an unweighted graph G, the Recursive Sparsest Cut algorithm achieves an $O(\alpha_n)$ approximation for the hierarchical clustering problem.
\end{theorem}
\begin{proof}
The proof follows easily by combining (\ref*{eq:one}), (\ref*{eq:two}), (\ref*{eq:three}), (\ref*{eq:four}) and summing over all clusters $A$ created by RSC.
See \hyperref[sec:omitted proofs]{Appendix~\ref*{sec:omitted proofs}}.
\end{proof}

\subsection{Generalized Cost Function and RSC}
In the original~\cite{dasguptaSTOC16} paper introducing the objective function of hierarchical clustering, Dasgupta also considered the more general cost function:
$cost_G(T)=\sum_{ij \in E} w_{ij} f(|\bold{leaves}(T[i\lor j])|)$, where $f$ is defined on the non-negative reals, is strictly increasing, and has $f(0) = 0$ (e.g. $f(x) = \ln(1 + x)$ or $f(x) = x^2$). For this more general cost function, he proved that a slightly modified greedy top-down heuristic (using $\tiny \dfrac{w(S,V\setminus S)}{\min((f|S|),f(|V\setminus S|))}$ with $\tiny \dfrac{1}{3}|V| \le |S|\le \dfrac{2}{3}|V|$, instead of Sparsest Cuts) continues to yield an
$O(\alpha_n\cdot \log n \cdot c_f)$ approximation\footnote{There isn't a direct polynomial time implementation of this heuristic for arbitrary functions $f$ to the best of our knowledge; however, a heuristic based on balanced cut will achieve similar guarantees.}, where $c_f\triangleq \max_{1\le n'\le n}\tiny \dfrac{f(n')}{f(n'/3)}$). Now, we analyze the previous RSC algorithm (with no modifications), but in the case of a weighted graph $G$ and when we are trying to optimize the generalized cost function. 

We again make the natural assumptions that the function $f$ acting on the number of leaves in subtrees, is defined on the nonnegative reals, is strictly increasing and $f(0) = 0$ (also see \hyperref[remark:cost-function]{Remark~\ref*{remark:cost-function}}). We also define: $c_f \triangleq \max_{1\le n'\le n} \dfrac{f(n')}{f(\lfloor n'/2 \rfloor)-f(\lfloor n'/4 \rfloor)}$. For what follows, we abuse notation slightly for ease of presentation and write $r/2, r/4$ etc. instead of $\lfloor r/2 \rfloor, \lfloor r/4\rfloor$ etc. As in the simple unweighted case, we use here the same definitions for $OPT$ and $E_{OPT}(t)$. Let $w(E_{OPT}(t))$ denote the total weight of the edges $E_{OPT}(t)$, i.e. the edges cut by $OPT$ at level $t$, where we define $w(\emptyset)= 0$ and we also define $g(t) \triangleq f(t+1) - f(t)$. We note that $\sum_{t=0}^{r-1} g(t) = f(r)-f(0)=f(r)$.

\begin{claim}
$\sum_{t=0}^{n-1} w(E_{OPT}(t))\cdot g(t) = OPT$ \label{claim:OPT}
\end{claim}

\begin{proof}
We will prove that the contributions of an edge $e=(u,v)$ to the LHS and RHS are equal.   
Let $A$ ($|A|=r_e$) be the minimal cluster in the optimal solution that contains both $u,v$. 
The contribution of $e$ to the RHS is: $w_e\cdot f(r_e)$. As for the contribution to the LHS, since $A$ is minimal and $|A|=r_e$, we deduce that $e\in OPT(t), \forall t<r_e$.
Also for levels $t \geq r_e$ we have $e\in A$ or some superset of $A$ and thus $e \not \in OPT(t)$
Hence the contribution to the LHS is: $w_e\cdot \sum_{t=0}^{r-1}g(t)=w_e\cdot f(r).$
\end{proof}

Focus on a cluster $A$ ($|A|=r$) in the solution produced by the algorithm.
Let $cut(A)$ denote the edges in $A$ cut by partitioning $A$.
This contributes
$w(cut(A))\cdot f(r)$ to the objective. 
We will charge our cost using the following quantity related to the optimum solution: 
$\sum_{t=r/4}^ {r/2-1}  w(E_{OPT}(t)\cap A) \cdot g(t)$. 

For that, we look at $OPT(r/2)\cap A$ and let's say that clusters $A_1, A_2,... ,A_k$ are induced by this partition, each being of size $|A_i|= \gamma_i |A| \le |A|/2=r/2$ ($\gamma_i \le 1/2$). 
Then,
\[ SC(A) \le \dfrac{\sum_i w(A_i,A\setminus A_i)}{r^2 \sum_i \gamma_i (1-\gamma_i)} \le \dfrac{2\cdot w(E_{OPT}(r/2)\cap A)}{r^2 \cdot 1/2}\]
where $SC(A)$ is the optimum sparsest cut (value) for $A$.
Since we used an $\alpha_n$-approximation, 
\begin{align*} 
\dfrac{w(cut(A))}{s(r-s)} \le \alpha_n\cdot SC(A) \le \dfrac{4\alpha_n \cdot (w(E_{OPT}(r/2)\cap A)}{r^2} \implies \\
w(cut(A)) \cdot f(r) \le 4\alpha_n \cdot \dfrac{s}{r} \cdot w(E_{OPT}(r/2)\cap A) \cdot f(r) \addtag \label{eq:gen-two}
\end{align*}
Since $w(E_{OPT}(t)\cap A) \geq w(E_{OPT}(t+1)\cap A)$, we have:
\begin{align*}
\sum_{t=r/4}^ {r/2-1}  w(E_{OPT}(t)\cap A) \cdot g(t) \ge w(E_{OPT}(r/2)\cap A) \cdot \sum_{t=r/4}^ {r/2-1}g(t) =\\
= (f(r/2)-f(r/4)) \cdot w(E_{OPT}(r/2)\cap A) \addtag \label{eq:gen-three}
\end{align*}
Using equations (\ref*{eq:gen-two}), (\ref*{eq:gen-three}), we get that:
\[ 
w(cut(A)) \cdot f(r) \le 4\alpha_n \cdot \dfrac{s}{r} \cdot \dfrac{f(r)}{f(r/2)-f(r/4)} \cdot \sum_{t=r/4}^ {r/2-1}  w(E_{OPT}(t)\cap A) \cdot g(t) \addtag \label{eq:gen-four}
\]
We now sum up the cost contributions of all clusters created in our hierarchical clustering solution. Let $s(A)$ be the size of the smaller piece produced in partitioning $A$.
\[ 
cost_{RSC}=\sum _A w(cut(A)) \cdot f(|A|) \leq 4\alpha_n \cdot c_f \sum_A \dfrac{s(A)}{|A|}\sum_{t=|A|/4}^ {|A|/2-1}  w(E_{OPT}(t)\cap A) \cdot g(t)
\addtag \label{eq:costRSC}
\] 
To complete our argument we need to make the comparison between OPT which is: $\sum_{t=0}^{n-1} w(E_{OPT}(t))\cdot g(t)$ and the sum
\[
\sum_A \dfrac{s(A)}{|A|}\sum_{t=|A|/4}^ {|A|/2-1}  w(E_{OPT}(t)\cap A) \cdot g(t) \addtag \label{eq:gen-five},
\]
where the first summation goes over all clusters $A$ in the solution we produce.
\begin{claim}
$\sum_A \dfrac{s(A)}{|A|}\sum_{t=|A|/4}^ {|A|/2-1}  w(E_{OPT}(t)\cap A) \cdot g(t) \le 2\cdot \sum_{t=0}^{n-1} w(E_{OPT}(t))\cdot g(t)$ \label{claim:better-SpCut}
\end{claim}
\begin{proof}
Consider some edge $e=(u,v) \in E_{OPT}(t)$. 
Focus on sets $A$ in the solution produced
such that $e \in E_{OPT}(t)\cap A$ so that $e$ contributes to the term $\sum_{t=|A|/4}^ {|A|/2-1}  w(E_{OPT}(t)\cap A) \cdot g(t)$ in the LHS. For all such clusters $A$, we need to have: $|A|/4 \le t < |A|/2 \implies 2t < |A| \le 4t$. 

Let $A_1,A_2,...,A_{k-1}$ be the sets for which the term $w(E_{OPT}(t)\cap A)$ appears: $A_1$ is the largest cluster (satisfying $2t < |A_1| \le 4t$) that contains the edge $e=(u,v)$ and when split we call its larger piece $A_2$ (again this set contains $e$) etc.,  $A_{k-1}$ is the last set for which the term appears and $(u,v)$ does not appear in $A_k$ ($A_k$ is the larger piece of the two that we got when we partitioned $A_{k-1}$). We have:
\begin{align*}
\sum_{i=1}^{k-1} \dfrac{s(A_i)}{|A_i|} = \dfrac{|A_1|-|A_2|}{|A_1|} + \dfrac{|A_2|-|A_3|}{|A_2|} + ... + \dfrac{|A_{k-1}|-|A_k|}{|A_{k-1}|}
\le \dfrac{\sum_{i=1}^{k-1} |A_i| - |A_{i+1}|}{\min_i |A_i|} \le \dfrac{|A_1|}{2t} \le 2. \addtag \label{eq:gen-six}
\end{align*}
(the constant can be optimized, but it does not change the asymptotic bound).

Thus the contribution of every edge $e \in E_{OPT}(t)$ to the LHS is at most $2 w_e g(t)$.
Note that this is exactly the contribution to the RHS.
This establishes the claim.

\end{proof}
\begin{theorem}
RSC achieves an $O(c_f \cdot \alpha_n)$ approximation of the generalized objective function for Hierarchical Clustering.
\end{theorem}

\begin{proof}
The proof follows from (\ref*{eq:costRSC}), \hyperref[claim:OPT]{Claim \ref*{claim:OPT}}, (\ref*{eq:gen-two}), (\ref*{eq:gen-three}), (\ref*{eq:gen-four}), (\ref*{eq:gen-five}), and (\ref*{eq:gen-six}).
See \hyperref[sec:omitted proofs]{Appendix~\ref*{sec:omitted proofs}}.
\end{proof}
 
\begin{remark} \label{remark:cost-function}
In order for our guarantee to be useful, we need $c_f$ to be a constant (or a slowly growing quantity). This would mean that $f$ is polynomially growing. We observe that in the case where the function $f$ is exponentially growing then our guarantee is not interesting (and in fact we may need to use a different strategy than RSC) and in the case $f$ is logarithmic, then we would get a factor $ \approx  O(\alpha_n \log n )$ approximation, which is the same guarantee as~\cite{dasguptaSTOC16}.
\end{remark}

\section{Hierarchical Clustering Hardness and the Small Set Expansion Hypothesis}\label{sec:hardness}

In this section, we prove a strong
inapproximability result, showing that, even in unweighted graphs (i.e. unit cost edges), the Hierarchical Clustering objective is hard to
approximate to within any constant factor, assuming the \textit{Small Set Expansion} hypothesis.

\subsection{$SSE$ and hardness amplification}

Given a graph $G(V,E)$, define the following quantities for non-empty subsets $S \subset V$:
normalized set size $\mu(S) \triangleq |S|/|V|$, and
edge expansion $\Phi_G (S) \triangleq ${\small $\dfrac{|E(S,V\setminus S)|}{\sum_{i \in S}d_i}$} (here $d_i$ is the degree of $i$).
The {\em Small Set Expansion} hypothesis was introduced by Raghavendra and Steurer~\cite{raghavendra2010graph}.

\begin{problem}[\SSE $(\eta,\delta)$]
Given a regular graph $G(V,E)$, distinguish between the following two cases:\\
\noindent{\bf Yes:} There exists a non-expanding set $S \subseteq V$ with $\mu(S) = \delta$ and $\Phi_G(S) \leq \eta$.\\
\noindent{\bf No:} All sets $S \subseteq V$ with $\mu(S) = \delta$ are highly expanding with $\Phi_G(S) \geq 1-\eta$.
\end{problem}

\begin{hypothesis}[Hardness of approximating \SSE]
For all $\eta >0$, there exists $\delta > 0$ such that the promise problem
\SSE $(\eta,\delta)$ is NP-hard.
\end{hypothesis}


\cite{raghavendra2010graph} showed that the Small Set Expansion Hypothesis implies the Unique Games Conjecture of Khot~\cite{khot2002power}. 
A decision problem is said to be SSE-hard if \SSE $(\eta,\delta)$ reduces to it by a polynomial time reduction for some constant $\eta$ and all $\delta > 0$.
Raghavendra, Steurer and Tulsiani~\cite{raghavendra2012reductions} showed the following hardness amplification result for graph expansion (see \hyperref[sec:prelim]{Preliminaries} for Gaussian Graphs definitions):
\begin{theorem} \label{th:amplification} 
For all $q \in \mathbb{N}$ and $\epsilon, \gamma >0$, it is SSE-hard to distinguish between the following two cases for
a given graph $H = (V_H, E_H)$:\\
\noindent\textbf{Yes:} There exist $q$ disjoint sets $S_1,..., S_q \subseteq V_H$ satisfying for all $l \in [q]$: $\mu(S_l)= 1/q$ and $\Phi_H(S_l) \le \epsilon + o(\epsilon)$.\\
\noindent\textbf{No:} For all sets $S\subseteq V_H$: $\Phi_H(S) \ge \Phi_{\mathcal{G}(1-\epsilon /2)}(\mu(S))- \gamma/\mu(S)$, where 
$\Phi_{\mathcal{G}(1-\epsilon /2)}(\mu(S))$ is the expansion of sets of volume $\mu(S)$ in the infinite Gaussian graph $\mathcal{G}(1-\epsilon /2)$.
\end{theorem}

\subsection{Hierarchical Clustering Hardness}

Now we are ready to prove our main hardness result. Our proof follows the argument of~\cite{raghavendra2012reductions} for establishing the hardness of {\sc Minimum Linear Arrangement}. We prove the following:

\begin{theorem} \label{hardness}

\textbf{(Hardness of Hierarchical Clustering).} 
For every $\epsilon > 0$, it is
$SSE$-hard to distinguish between the following two cases for a given graph $G = (V, E)$, with $|V| = n$:\\
\textbf{Yes:} There exists a decomposition tree $T$ of the graph such that 
$cost_G(T)\le \epsilon n |E|$\\
\textbf{No:} For any decomposition tree $T$ of the graph 
$cost_G(T) \ge c\sqrt{\epsilon}n |E|$.

\end{theorem}

\begin{proof} 
We apply \hyperref[th:amplification]{Theorem \ref{th:amplification}} for the following values: $q = \lceil 2/\epsilon \rceil, \epsilon' = \epsilon/3$ and $\gamma = \epsilon$.
We need to first handle the \textbf{Yes} case.
We get that the vertices can be divided into sets $S_1, S_2, ..., S_q$, each having size $n/q = n\epsilon/2$, such that at most $\epsilon' + o(\epsilon')$ fraction of edges leave the sets (i.e. go across sets). Now consider the hierarchical clustering solution that first partitions the vertices into the sets $S_1, S_2, ..., S_q$ and then partitions each $S_i$ arbitrarily. 
Edges inside the set $S_i$ contribute at most $|S_i|$ to the objective function 
and this is $|S_i|=n/q=\epsilon n/2$. Moreover, edges whose endpoints are in different sets will have contribution at most $n$; but there are at most $\epsilon/2$ fraction of such edges and so the overall objective for this hierarchical clustering solution is at most $\epsilon n |E|$. 

Now, we handle the \textbf{No} case by using the argument of~\cite{raghavendra2012reductions} for \MLA \ that follows from an observation of~\cite{devanur2006integrality} and the fact that the objective function of \MLA \ is always less than the cost of Hierarchical Clustering. To see the latter, observe that if we have a hierarchical clustering tree $T$ then consider the ordering of the vertices induced by the order that they appear as leaves in $T$ (like projecting the leaves to a line). Then, the stretch of an edge $(u,v)$ that is cut, can be at most the size of the subtree that corresponds to that edge and this is exactly the quantity: $|\textbf{leaves}(T[u\lor v])|$. Since we know (\cite{raghavendra2012reductions, devanur2006integrality}) that in the \textbf{No} case, for all orderings $\pi: V \rightarrow [n], \mathbb{E}_{(u,v)\sim E} [|\pi(u)-\pi(v)|]     \ge c\sqrt{\epsilon}n$, it immediately follows that:
$cost_G(T) \ge c\sqrt{\epsilon}n |E|$. 

\end{proof}

\section{Approximation for HC using SDP} \label{sec:sdp-approx}


In this section, we present our SDP relaxation for HC based on spreading metrics, we point out its relation with the SDP relaxation of $k$-\textit{balanced partitioning} in~\cite{krauthgamer2009partitioning} and we prove that it is an $O(\sqrt{\log n})$ approximation for both the simple and the generalized cost function.

\subsection{Writing the SDP}
We view a hierarchical clustering of $n$ data points as
a collection of partitions of the data, one for each level $t = n-1,\ldots,1$.
The partition for a particular level $t$ satisfies the property that every cluster has size at most $t$; additionally, for every vertex $i$, the cluster containing vertex $i$ at level $t$ is the maximal cluster in the hierarchy with size at most $t$.
The partition at level ($t-1$) is a refinement of the partition at level $t$.
Note that the partition corresponding to $t=1$ must consist of $n$ singleton clusters.
We represent the partition at level $t$ by the set of variables
$x_{ij}^t$, $i,j \in V$, where $x_{ij}^t = 1$ if $i$ and $j$ are in different clusters in the partition at level $t$ and $x_{ij}^t = 0$ if $i$ and $j$ are in the same cluster.
We point out some properties of these variables 
$x_{ij}^t$ satisfied by an integer solution corresponding to an actual hierarchical clustering:
\begin{enumerate}
\item {\bf refinement:} $x_{ij}^t \leq x_{ij}^{t-1}$. If $i$ and $j$ are separated at level $t$, then they continue to be separated at level $t-1$.
\item {\bf triangle inequality:} $x_{ij}^t + x_{jk}^t \geq x_{ik}^t$. In the clustering at level $t$, if $i$ and $j$ are the same cluster, $j$ and $k$ are in the same cluster, then $i$ and $k$ are in the same cluster.
\item {\bf $\ell_2^2$ metric:} The triangle inequality condition implies that $x_{ij}^t$ is a metric. Further, we can associate unit vectors $v_i^t$ with vertices $i$ at level $t$ such that
$x_{ij}^t = \frac{1}{2}||v_i^t - v_j^t||_2^2$. In order to do this, all vertices in the same cluster at level $t$ are assigned the same vector, and vertices in different clusters are assigned orthogonal vectors.
\item {\bf spreading:} $\sum_j x_{ij}^t \geq n-t$. For the clustering at level $t$, there are at most $t$ vertices in the same cluster as $i$. Hence there are at least $n-t$ vertices in different clusters. For each such vertex $j$, $x_{ij}^t = 1$ implying the inequality.
\item {\bf cluster size:} The size of the smallest cluster in the hierarchy containing both vertices $i$ and $j$ is given by $1+\sum_{t=1}^{n-1} x_{ij}^t$. Suppose $C$ is the smallest cluster containing both $i$ and $j$. Then for $t \geq |C|$, the partition at level $t$ must contain $C$ or some superset of $C$. Hence $x_{ij}^t=0$ for $t \geq |C|$. For $t < |C|$, the clustering at level $t$ must have $i$ and $j$ in different clusters, hence $x_{ij}^t = 1$.
Hence $\sum_{t=1}^{n-1} x_{ij}^t = |C|-1$. Finally, we can write the SDP relaxation SDP-HC as follows:
\end{enumerate}
\[
\min \sum_{t=1}^{n-1} \sum_{ij \in E} x_{ij}^tw_{ij}= \min \sum_{t=1}^{n-1} \sum_{ij \in E} \dfrac{1}{2} \lVert v_i^t-v_j^t \rVert_2^2w_{ij}  \tag{SDP-HC} 
\]
\[
\text{\ such that: \ \ } x_{ij}^{t} \le x_{ij}^{t-1}, \mbox{\ \ \ } t=n-1,n-2,...1\\ \]
\[
x_{ij}^{0} =1, \forall i,j \in V \ \text{and}\ x_{ij}^t \le 1, \forall i,j,t \]
\[x_{ij}^t =\dfrac{1}{2} \lVert v_i^t-v_j^t \rVert_2^2 \ \text{and}\ \lVert v_i^t \rVert_2^2 =1, \forall i \in V\\ \]
\[x_{ij}^t \le x_{jk}^t+ x_{ik}^t, \forall i,j,k \in V, \ \forall t \ \text{and}\ \sum_{j}x_{ij}^t \ge n-t, \forall i,t  \]

It is easy to see that an optimal solution to SDP-HC can be computed in polynomial time.
By the preceding discussion, we have shown that SDP-HC is a valid relaxation for HC:

\begin{lemma}
The value of an optimal solution to SDP-HC can be computed in polynomial time, and gives  a lower bound on the cost of an optimal solution to the hierarchical clustering problem.
\end{lemma}

\subsection{Connections of SDP-HC with Balanced Partitioning}
The authors of~\cite{krauthgamer2009partitioning} write an SDP relaxation for the problem of $k$-Balanced Partitioning ($k$-BP) which was the following (SDP-$k$-BP):
\[
\min \sum_{ij \in E} w_{ij}\cdot \dfrac{1}{2}\lVert v_i-v_j\rVert_2^2 \tag{SDP-$k$-BP} \label{SDP-$k$-BP}
\]
\[ \text{\ such that:\ \ } \lVert v_i-v_j \rVert_2^2+\lVert v_j-v_k \rVert_2^2 \ge \lVert v_i-v_k \rVert_2^2, \ \forall i,j,k \in V\\ \]
\[\sum_{j\in S} \dfrac{1}{2}\lVert v_i-v_j\rVert_2^2 \ge |S|-\dfrac{n}{k}, \ \forall S\subseteq V, i \in S \]

\noindent Their result was that the above relaxation is an $O(\sqrt{\log k \log n})$ approximation (bi-criteria $\nu=2$) algorithm for $k$-BP, that will create pieces of size at most $2n/k$.

\begin{claim}
Let $A$ be a cluster of size $r$. SDP-HC solution restricted to set $A$, at level $t=r/4$ is a valid solution for $k$-balanced partitioning based on the \ref*{SDP-$k$-BP} relaxation, where $k=4$.
\end{claim}

\begin{proof}
See \hyperref[sec:omitted proofs]{Appendix \ref*{sec:omitted proofs}}.
\end{proof}

In order to produce a hierarchical clustering from the SDP solution, we recursively partition $V$ in a top down fashion: while partitioning a cluster $A$, we use the SDP-HC solution restricted to set $A$ at level $t=|A|/4$ as a valid solution for $4$-balanced partitioning and invoke the algorithm of \cite{krauthgamer2009partitioning} as a black box.
Let $E_A$ be the edges cut by the algorithm when splitting cluster $A$. From the analysis of \cite{krauthgamer2009partitioning} , 
we get that (for us $k=4$, so $\log k$ is constant):
\[ 
w(E_A) \le O(\sqrt{\log n})\cdot SDP_A(r/4) \addtag \label{eq:KNS}
\]
and we partition $A$ into pieces of size at most $\le 2\cdot r/4=r/2$ (bi-criteria).
In the analysis that follows we will use this result as a black box.

\subsection{$O(\sqrt{\log n})$ approximation for Hierarchical Clustering}

Now we go on to see that the integrality gap of our SDP-HC is $O(\sqrt{\log n})$. 
%
Let $r$ be the size of a cluster $A$ in the solution produced. For our charging argument, we observe that the we pay $r\cdot w(E_A)$ where $E_A$ are the edges cut by the \cite{krauthgamer2009partitioning} algorithm when partitioning $A$. We will charge this cost to $\sum_{t=r/8+1}^{r/4}SDP_A(t)\ge \dfrac{r}{8} SDP_A(r/4)$ (note that as $t$ decreases more edges are cut). Thus, using~\cite{krauthgamer2009partitioning}, the total cost of the solution produced (where $r$ depends on $A$):

\[
cost_{HC} = \sum_{A} r\cdot w(E_A)\le O(\sqrt{\log n}) \sum_A \sum_{t=r/8+1}^{r/4}SDP_A(t) \addtag \label{eq:opt-sdp}.
\]

\begin{claim}
$\sum_{A} \sum_{t=|A|/8+1}^{|A|/4} SDP_A(t) \le O$(SDP-HC). \label{claim:sdp-claim}
\end{claim}

\begin{proof}
See \hyperref[sec:omitted proofs]{Appendix \ref*{sec:omitted proofs}}.
\end{proof}
\begin{theorem}
The cost of the solution produced by the SDP-HC rounding algorithm is within a factor of $O(\sqrt{\log n})$ from the SDP value.
\end{theorem}
\begin{proof}
Using \hyperref[claim:sdp-claim]{Claim~\ref{claim:sdp-claim}} and (\ref*{eq:opt-sdp}) we get that $cost_{OPT}\le O(\sqrt{\log n})\cdot$SDP-HC.
\end{proof}

\subsection{The case of the generalized cost function}

Now, we consider the performance of SDP-HC-gen (where SDP-HC-gen is essentially the same as SDP-HC where actually we multiply each term in the objective function by $g(t)=f(t+1)-f(t)$) for the generalized cost function and we show essentially the same guarantee (for the proof, see \hyperref[sec:omitted proofs]{Appendix \ref*{sec:omitted proofs}}):
\begin{theorem}
The cost of the solution produced by the SDP-HC-gen rounding algorithm is within a factor of $O(\sqrt{\log n}\cdot c_f)$ from the 
SDP value where $c_f \triangleq \max_{r\in \{1,...,n\}}\tiny \dfrac{f(r)}{f(r/4)-f(r/8)}$.
\end{theorem}




\section{Conclusion and Further Research}\label{sec:conclusion}
We proved that the recently introduced objective function for hierarchical clustering in~\cite{dasguptaSTOC16}, can be approximated within a factor of $O(\alpha_n)$ by repeatedly taking ($\alpha_n$ approximations to) Sparsest Cuts and within $O(\sqrt{\log n})$ using a spreading metric SDP relaxation. We also proved that it is hard to approximate the HC objective function within any constant factor assuming the \textit{Small Set Expansion Hypothsesis}, which was the first strong inapproximability for the problem to the best of our knowledge. We finally presented an LP based $O(\log n)$ approximation by showing that HC falls into the spreading metrics paradigm of~\cite{even2000divide}.
 
We would like to conclude the paper asking if we can do even better for this particular problem. The reason why we might face difficulties in improving the approximation guarantee may have to do with the much more basic problem of Minimum Balanced Bisection. It seems implausible that we would get a better approximation for hierarchical clustering before getting an improvement in the current best approximation guarantee for Balanced Bisection. It is also interesting to try to come up with other suitable cost functions for HC, apart from those considered here.

Another direction for research is that of \textit{beyond worst case analysis}. What can we say about exact recovery on $\gamma$-stable instances under Bilu-Linial~\cite{bilu2012stable} notion of stability? For example, in~\cite{makarychevSODA14} they show that the standard SDP relaxation
for \textsc{Max-Cut} is integral if the instance is sufficiently stable ($\gamma \ge c\sqrt{\log n}\log n$ for some
absolute constant $c>0$). It would be nice to formalize and say something similar for our problem, since this would not only explain the success of certain heuristics for HC based on finding sparsest cuts, but also justify their use in practice, assuming that in real applications most instances are stable (such an assumption for clustering problem is widely accepted; for more see~\cite{balcan2012clustering, bilu2012stable, makarychevSODA14} and references therein). Finally, we also find interesting the scenario where the input graph is drawn from a probability distribution for which there is a truly hierarchical structure. Can we then prove that a suitable SDP relaxation will indeed find a hierarchical structure close to the actual underlying hierarchy with high probability?


\clearpage
\bibliographystyle{alpha}
\bibliography{main}

\clearpage

\appendix
\section{\Large Preliminaries \label{sec:prelim}}

Here, we would like to briefly discuss some important problems and definitions that will frequently come up in the rest of the paper. Some additional definitions and facts may be presented in the sections for which they are relevant.

\SC. Given a weighted, undirected graph $G=(V,E,w)$ ($|V|=n$) we want to find a set $S\neq \emptyset, V$ that minimizes the ratio:
\[
\dfrac{w(S,V\setminus S)}{|S|\cdot |V\setminus S|}
\]
It is an NP-hard problem for which many important results are known including the LP relaxation of Leighton-Rao~\cite{leighton1999multicommodity} with approximation ratio $O(\log n)$ and the SDP relaxation with triangle
inequality of Arora, Rao, Vazirani~\cite{arora2009expander} with approximation ratio $O(\sqrt{\log n})$; it is a major open question if we can improve this approximation ratio.

\SSE. SSE is a hardness assumption that informally tells us the following: Given a graph $G$, it should be hard to distinguish between the case where there exists a small set $S$ that has only a few edges leaving it versus the case where for all small sets $S$ there are many edges leaving the sets. For a formal statement see \hyperref[sec:hardness]{Section \ref{sec:hardness}}. This hardness assumption is closely connected to the Unique Games Conjecture (UGC) of \cite{khot2002power} and its variants. In particular, the SSE Hypothesis implies UGC(\cite{raghavendra2010graph}) and it has been used to prove many inapproximability results for problems like balanced separator and minimum linear arrangement (\cite{raghavendra2012reductions}).

\kBP. Given a weighted undirected graph $G$ on $n$ vertices, the goal is to partition the vertices into $k$ equally sized components of size roughly $n/k$ so that the total weight of the edges connecting different components is small.
It is an important generalization of well-known graph partitioning problems, including minimum bisection ($k$=2) and minimum balanced cut and it has applications in VLSI design, data mining (clustering), social network analysis etc. It is an NP-hard problem and the authors of~\cite{krauthgamer2009partitioning} present a bi-criteria (which means that pieces may have size $2n/k$ rather than $n/k$) approximation algorithm achieving an approximation of $O(\sqrt{\log n \log k})$. Their result will be useful in our analysis for our spreading metrics SDP in \hyperref[sec:sdp-approx]{Section \ref{sec:sdp-approx}}. However, for us the dependence on $k$ will be unimportant since in our analysis we only need $k$ to be a small constant (e.g. $k$=4).

\MLA. Given a weighted undirected multigraph $G(V,E,w)$ ($|V|=n$) we want to find a permutation $\pi: V\rightarrow \{1,2,\dots,|V| \}$ that minimizes:
\[
\sum_{(x,y)\in E, x<y}w(x,y) \cdot | \sigma(y)-\sigma(x) |
\]
A factor $O(\sqrt{\log n}\log \log n)$ approximation for MLA was shown in~\cite{charikar20062,feige2007improved}. In addition, some recent hardness results are also known: in \cite{raghavendra2012reductions} it is shown that it is SSE-hard to approximate MLA to within any fixed constant factor and in \cite{ambuhl2011inapproximability} the authors prove that MLA has no polynomial time approximation scheme, unless NP-complete problems can be solved in randomized subexponential time.

\textsc{Gaussian Graphs.}
For a constant $\rho \in (-1, 1)$, let $\mathcal{G}(\rho)$ denote the infinite graph over $\mathbb{R}$ where the weight
of an edge $(x, y)$ is the probability that two standard Gaussian random variables $X, Y$ with correlation $\rho$ equal
$x$ and $y$ respectively. The expansion profile of Gaussian graphs is given by $\Phi_{\mathcal{G}(\rho)}(\mu) = 1 - \Gamma_{\rho}(\mu)/\mu$ where
the quantity $\Gamma_{\rho}(\mu)$ defined as
\[
\Gamma_{\rho}(\mu) \triangleq \mathbb{P}_{(x,y)\sim \mathcal{G}_\rho} (x\ge t, y\ge t),
\]
where $\mathcal{G}_{\rho}$ is the 2-dimensional Gaussian distribution with covariance matrix:
\[
\begin{bmatrix}
1 & \rho \\
\rho & 1
\end{bmatrix}
\]
and $t\ge 0$ is such that $\mathbb{P}_{(x,y)\sim \mathcal{G}_{\rho}} \{x \ge t\} = \mu$.

\section{Spreading Metrics and Bartal's Decomposition}\label{sec:spreading}

Bartal (\cite{bartal2004graph}) presented a graph decomposition lemma and used it as a key ingredient in order to prove an $O(\log n)$ approximation guarantee for the spreading metrics paradigm in undirected graphs, thus improving the results for many problems considered in \cite{even2000divide}. At a high level, the decomposition finds a cluster in the graph that has a low diameter, such that the weight of the cut created is small with respect to the weight of the cluster. The decomposition is essentially based on the decomposition of \cite{garg1993approximate} performed in a careful manner so as to achieve a more refined bound on the ratio between the cut and the cluster's weight.

Let $G=(V,E)$ be an undirected graph with two weight functions $w,l : E\rightarrow \mathbb{R}^+$. We interpret $l(e)$ to be the length of the edge $e$, and the distance $d(u,v)$ between pairs of vertices $u,v$ in the graph, is determined by the length of the shortest path between them. Given a subgraph $H=(V_H,E_H)$ of $G$, let $d_H$ denote the distance in $H$, let $\Delta(H)$ denote the diameter of $H$, and $\Delta=\Delta(G)$. We also define the
volume of $H$, $\phi(H)=\sum_{e \in E_H} w(e)l(e)$.

Given a subset $S\subseteq V$, $G(S)$ denotes the subgraph of $G$ induced by $S$. Given partition $(S,\bar{S})$, let $\Gamma(S)= \{(u,v) \in E ; u \in S, v \in \bar{S} \}$ and $cut(S) = \sum_{e \in \Gamma(S)} w(e)$.
For a vertex $v$ and $r \ge 0 $, the ball at radius $r$ around $v$ is defined as $B(v,r)= \{u \in V | d(u,v) \le r \}$. Let $S=B(v,r)$. Define

\begin{align*}
\bar{\phi}(S)=\bar{\phi}(v,r)=\sum_{e=(u,w): u,w \in S } w(e)l(e) + \sum_{e=(u,w) \in \Gamma(S)} w(e)(r - d(v,u)).
\end{align*}

Given a subgraph $H$, we can similarly define $\bar{\phi}_H$ with respect to the subgraph $H$. Define the spherical-volume of $H$,
\begin{align*}
\phi^*(H)= \max_{v\in H} \bar{\phi}_H(v,\dfrac{\Delta(H)}{4}).
\end{align*}

In the following, we state three basic lemmas, which are based on a standard argument similar to that of \cite{garg1993approximate}, before stating the main result in \ref{thm:spreading}. For the proofs, we refer the reader to \cite{bartal2004graph}. 

\begin{lemma}
Given a graph $G$, there exists a partition $(S,\bar{S})$ of $G$, where $S$ is a ball and 
\begin{align*}
cut(S)\le \dfrac{8\ln(\phi^*(G)/\phi^*(G(S)))}{\Delta(G)} \cdot \bar{\phi(S)}.
\end{align*}
\end{lemma}

The decomposition will become useful when it is applied recursively. This is particularly important for our main application which is hierarchical clustering and we can have a recursive approach where we find a good cut, creating two subgraphs and then we continue with the two new components. More generally, for the second lemma of \cite{bartal2004graph}, we will be interested in applications which are associated with a cost function $cost$ over subgraphs $\hat{G}$ of $G$
which is nonnegative, 0 on singletons and obeys the following natural recursion rule:

\begin{align}
cost(\hat{G}) \le cost(\hat{G}(S)) + cost(\hat{G}(\bar{S})) + \Delta(\hat{G}) \cdot cut(S).  \label{eq:bartal-lemma}
\end{align}

Now we state the second basic lemma:

\begin{lemma}
Any cost function defined by (\ref*{eq:bartal-lemma}) obeys $cost(G)\le O(\log (\phi/\phi_0)) \cdot \phi(G)$, where $\phi=\phi(G)$ and $\phi_0$ is the minimum value of $\phi(\hat{G})$ on non-singleton subgraphs $\hat{G}$.
\end{lemma}

Finally, we have the third lemma which will give us the $O(\log n)$ approximation. We can obtain a bound depending only on $n$, by modifying the process slightly by associating a volume $\phi(G)/n$ with the nodes, like in \cite{garg1993approximate}. This will ensure that $\phi_0 \ge \phi(G)/n$ and by substituting we get what we want:

\begin{lemma}
The function defined by (\ref*{eq:bartal-lemma}) using the modified procedure obeys $cost(G) \le O(\log n) \cdot \phi (G)$.
\end{lemma}

Now we turn our attention to the connection with the spreading metrics paradigm. 
Having the definition of a spreading metric in mind (see \hyperref[sec:lp-approx]{Section~\ref{sec:lp-approx}}) and the previous three recursive graph decomposition lemmas we can easily obtain the following theorem, as proved in \cite{bartal2004graph}:

\begin{theorem}\label{thm:spreading}
There exists an $O(\log n)$ approximation for problems in the spreading metrics paradigm.
\end{theorem}

\section{LP Spreading Metrics and $O(\log n)$ approximation}\label{sec:lp-approx}

We prove here that the hierarchical clustering objective function defined above falls
into the \textit{divide and conquer approximation algorithms via spreading metrics} paradigm of \cite{even2000divide}. 

The spreading metric paradigm applies to minimization problems on undirected graphs $G=(V,E)$ with edge weights $w(e)\ge 1$. We also have an auxiliary graph $H$ and a scaler function on subgraphs of $H$ (e.g. size of the components of $H$). A decomposition tree $T$ is a tree with nodes corresponding to non-overlapping subsets of $V$, forming a recursive partition of the nodes $V$. For a node $t$ of $T$, we denote by $V_t$ the subset at $t$. Associated are the subgraphs $G_t, H_t$ induced by $V_t$. Let $F_t$ be the set of edges that connect vertices that belong to different children of $t$, and $w(F_t)= \sum_{e\in F_t}w(e)$. The cost of $T$ is $cost(T)= \sum_{t\in T} scaler(H_t)\cdot w(F_t)$.

\begin{definition}
A spreading metric is a function on the edges of the graph $l : E\rightarrow \mathbb{R}^+$ satisfying the following two properties:
\begin{enumerate}
\itemsep=0pt
\item Lower bound property: 
The volume of the graph $\sum_{e\in E}w(e)l(e)$ is a lower bound on the optimal cost.
\item Diameter property: 
For any $U \subseteq V$ and $H_U$ the subgraph of $H$ induced by $U$, has diameter $\Delta(U)\ge scaler(H_U)$.
\end{enumerate}
\end{definition}

We closely follow their formulation for the Linear Arrangement problem, which also falls into the spreading metrics paradigm, but we make the necessary semantic changes.
We need to show the divide and conquer applicability and the spreading metrics applicability of their result for our problem.

Firstly, to establish the divide and conquer
applicability we consider any binary decomposition tree $T$ that fully
decomposes the problem.(we normalize the edge weights by dividing with the minimum edge weight). Note that there is a $1-1$ correspondence between the
leaves of $T$ and the vertices of $G$. The solution to the hierarchical clustering problem that is represented by $T$ is easily given by the cuts, in $G$, induced by the internal nodes of $T$. The cost of the tree $T$ is:
\begin{align}
cost_G(T)=\sum_{t \in T}|V_t|w(F_t).
\end{align}
where $V_t$ and $F_t$ are the set of vertices and cut corresponding
to the tree node $t$ and $w(F_t)$ is the total weight of the edges cut at this internal node $t$.
We need to show that this cost bounds the cost of
solutions built up from $T$. For this we prove that for every tree node $t$ the cost of
the subtree rooted at $t$, denoted $T_t$, bounds the cost of solutions built up from $T_t$
to the hierarchical clustering problem for the subgraph of G induced by the set of
vertices $V_t$. We prove the claim by induction on the level of the tree nodes. The 
claim clearly holds for all leaves of $T$. Consider an internal tree node $t \in T$ and
denote its two children by $t_L$ and $t_R$. By induction the claim holds for both $t_L$ and
$t_R$. The solution represented by $T_t$ is given by concatenating the solutions
represented by $T_{t_L}$ and $T_{t_R}$. Note that the additional cost is at most $|V_t|$ times the capacity of the cut $F_t$ that separates $V_{t_L}$ from $V_{t_R}$. We get
\begin{align}
cost_G(T_t) \leq cost_G(T_{t_L}) + cost_G(T_{t_R})+ |V_t|w(F_t).
\end{align}
The inductive claim follows.

We now show how to compute the spreading metric that assigns length $l(e)$ to an edge $e \in E$ of the graph. Consider the following
linear program (LP1):

\begin{align}
\min \sum_{e \in E} w(e)\cdot l(e)\\
s.t.\ \ \ \ \  \forall U \subseteq V, \forall v \in V: \sum_{u \in U} dist_l(u,v) \geq \dfrac{1}{4} (|U|^2-1) \label{ineq:spreading}\\
\forall e \in E: l(e) \geq 0
\end{align}
In the linear program, we follow the notation that regards $l(e)$ as edge
lengths, and $dist_l(u, v)$ is the length of the shortest path from $u$ to $v$. We will refer to constraint (6) as the spreading constraint.
The linear program can be solved in polynomial time since we can construct a separation oracle.
In order to verify that the spreading constraint (\ref*{ineq:spreading}) is satisfied, for each vertex $v$, we sort the vertices in $V$ in increasing order of distance $dist_l(u,v)$ and verify the spreading constraint for all prefixes $U$ of this sorted order.

\begin{lemma*}
Let $l(e)$ denote a feasible solution of the linear program. For every $U \subseteq V$ with $(|U|>1)$, and for every vertex $v \in U$ there is a vertex $u \in U$ for which $dist_l(u,v)\geq \dfrac{1}{10}|U|$.
\end{lemma*}

\begin{proof}
The average distance of a node $u\in U-\{v\}$ from $v$ is greater than $\dfrac{1}{4} (|U|-1)$, 
because of the constraint corresponding to $U$ and $v$. Therefore,
there exists a vertex $u \in U$ whose distance from $v$ is at least the average distance
from $v$, and the lemma follows, since $dist_l(u,v)\geq \dfrac{1}{4} (|U|-1)\geq \dfrac{1}{10}|U|$. $(|U|>1)$
\end{proof}

Note that the previous lemma comes short of the diameter guarantee by a
factor of 10: while the diameter guarantee requires that the diameter of a subset
$U$ be greater than $|U|$, the proven bound is only $|U|/10$. However, it is known that this only affects the constant in the approximation
factor.

In the next lemma, we prove that the volume of an optimal solution of the
linear program satisfies the lower bound property.

\begin{lemma*}
The cost of an optimal solution of the linear program is a lower
bound on the cost of an optimal hierarchical clustering of $G$.
\end{lemma*}

\begin{proof}
Consider any binary hierarchical clustering given by the sequence of cuts in the decomposition tree $T$ and define $l(e)= |\bold{leaves}(T[i\lor j])|$ for edge $e=(i,j) \in E$. It is easy to see that this is indeed a metric and it is actually an 
ultrametric. We show that $l(e)$ is a feasible solution for the linear program. The cost $\sum_{e\in E} w(e)\cdot l(e)$ equals the cost of the hierarchical clustering induced by the tree $T$. The feasibility of $l(e)$ is proved as follows: Consider a subset $U \subseteq V$ and a vertex $v \in U$. We observe that the average distance from $v$ of the vertices in $U$ will be minimized when $U$ is ``packed around" $v$, meaning that with each cut we peel off only one vertex at a time. We have that:
\begin{align*}
\sum_{u \in U}dist_l(u,v)=2+3+...+|U| \geq \dfrac{1}{2}(|U|^2-1) \geq \dfrac{1}{4}(|U|^2-1)
\end{align*}

Hence, $l(\cdot)$ is a feasible solution and the lemma follows.
\end{proof}

With the above two lemmas we have proved that our hierarchical clustering objective function falls into the spreading metrics paradigm, because it satisfies the lower bound property and the diameter property. 
Using Bartal's decomposition and specifically \hyperref[thm:spreading]{Theorem \ref{thm:spreading}} from \hyperref[sec:spreading]{Section \ref{sec:spreading}} we get an approximation guarantee of $O(\log n)$:

\begin{theorem} 
There exists an $O(\log n)$ approximation for the hierarchical clustering objective function defined by (\ref*{obj:clustering}).
\end{theorem}

\section{Omitted Proofs}\label{sec:omitted proofs}

\begin{theorem}
Given an unweighted graph G, the Recursive Sparsest Cut algorithm achieves an $O(\alpha_n)$ approximation for the hierarchical clustering problem.
\end{theorem}
\begin{proof}
By combining (\ref*{eq:one}), (\ref*{eq:two}), (\ref*{eq:three}), (\ref*{eq:four}) and summing over all clusters $A$ created by RSC, we get the following result for the overall performance guarantee:

\[ cost_{RSC} = \sum_A r \cdot |E(B_1,B_2)| \le \sum_A 4\alpha_n s |E_{OPT}(\lfloor r/2 \rfloor) \cap A| \le \] \\
\[ \le 4\alpha_n  \sum_A \ \sum_{t=r-s+1}^r |E_{OPT}(\lfloor t/2 \rfloor) \cap A| \le 4\alpha_n \sum_{t=1}^n |E_{OPT}(\lfloor t/2 \rfloor)| \le 8\alpha_n \cdot OPT\]

\end{proof}

\begin{theorem}
RSC achieves a $O(c_f \cdot \alpha_n)$ approximation of the generalized objective function for Hierarchical Clustering.
\end{theorem}

\begin{proof}
From (\ref*{eq:costRSC}),
\[ 
cost_{RSC} \leq 4\alpha_n \cdot c_f \sum_A \dfrac{s(A)}{|A|}\sum_{t=|A|/4}^ {|A|/2-1}  w(E_{OPT}(t)\cap A) \cdot g(t)
\] 
Combining the above with \hyperref[claim:OPT]{Claim \ref*{claim:OPT}}, (\ref*{eq:gen-two}), (\ref*{eq:gen-three}), (\ref*{eq:gen-four}), (\ref*{eq:gen-five}), (\ref*{eq:gen-six}), we get that the total cost of the RSC is at most $cost_{RSC} \le (8 c_f \alpha_n) \cdot OPT =  O(c_f \cdot \alpha_n)$.
\end{proof}

\begin{claim}
Let $A$ be a cluster of size $r$. SDP-HC solution restricted to set $A$, at level $t=r/4$ is a valid solution for $k$-balanced partitioning based on the \ref*{SDP-$k$-BP} relaxation, where $k=4$.
\end{claim}

\begin{proof}
To see this we need to compare the set of constraints imposed by SDP-HC and \ref*{SDP-$k$-BP}.
In SDP-HC, we have some additional constraints: $x_{ij}^t\le 1$ and $v_i^t \le 1$, but that is fine since imposing extra constraints just makes a stricter relaxation.
Now let's look at the spreading constraints: In SDP-HC we have $\sum_{j}x_{ij}^t \ge n-t \implies \sum_{j \in S}x_{ij}^t \ge |S|-t$ which is basically the \ref*{SDP-$k$-BP} spreading constraints. Thus, by looking at the SDP-HC solution restricted to set $A$ ($|A|=r$), at level $t=r/4$, we can get a valid 4-balanced partitioning solution  of $A$.
\end{proof}

\begin{claim}
$\sum_{A} \sum_{t=|A|/8+1}^{|A|/4} SDP_A(t) \le O$(SDP-HC). \label{claim:sdp-claim}
\end{claim}

\begin{proof}
The flavor of this analysis is similar to our RSC result from \hyperref[sec:better-SpCut]{Section \ref*{sec:better-SpCut}}. Let's look at an edge $e=(u,v)$ at a fixed level $t$. For which sets $A$ do we get the term SDP$_A(t)$ where both endpoints $u,v \in A$? In order for $u,v$ to belong to $A$: $t\in \big (|A|/8, |A|/4 \big ] \implies 4t\le |A| < 8t$. There can be at most one such $|A|$, so LHS is charged only once. To see why $A$ is unique, suppose we had two such clusters $|A_1|, |A_2|$ that both contained $u,v$ with their sizes $|A_1|, |A_2| \in [4t,8t)$. Since we have a hierarchical decomposition, one of $A_1,A_2$ is ancestor of the other. Let's say, wlog, $A_1$ is ancestor of $A_2$. But then, all of its descendants are of size below the range $[4t,8t)$ due to the 4-partition, which is a contradiction.
\end{proof}

\begin{remark}
In the above analysis, whenever we write $|A|/4$ we mean $\lfloor |A|/4 \rfloor$. However this will not affect the result. Plus, we used $O$(SDP-HC), because some additional constants might be introduced whenever the set $A$ is small ($|A|<8$). 
\end{remark}

\begin{theorem}
The cost of the solution produced by the SDP-HC-gen rounding algorithm is within a factor of $O(\sqrt{\log n}\cdot c_f)$ from the 
SDP value where $c_f \triangleq \max_{r\in \{1,...,n\}}\tiny \dfrac{f(r)}{f(r/4)-f(r/8)}$.
\end{theorem}
\begin{proof}

Let A be a cluster of size $|A|=r$ and let $g(t)=f(t+1)-f(t)$. We want to compare the cost of OPT for splitting $A$ with our solution SDP$_A(t)$ for levels $t = r/8 +1,...,r/4$. Using (\ref*{eq:KNS}):
\[
cost_{OPT}(A) = f(r)\cdot w(E_A)\le O(\sqrt{\log n}) f(r)\cdot SDP_A(r/4)\le \]
\[ \le O(\sqrt{\log n}) \dfrac{f(r)}{f(r/4)-f(r/8)} \sum_{t=r/8+1}^{r/4}SDP_A(t)\cdot g(t)\le \]
\[ \le O(\sqrt{\log n})\cdot c_f \sum_{t=r/8+1}^{r/4}SDP_A(t)\cdot g(t).\]

Using now \hyperref[claim:sdp-claim]{Claim \ref*{claim:sdp-claim}} and summing over all clusters $A$ in the hierarchical clustering we get:

\[ 
OPT\le O(\sqrt{\log n})\cdot c_f \cdot \text{SDP-HC}.
\]

where $\sum_A \sum_{t=r/8+1}^{r/4} \text{SDP}_A(t)\cdot g(t) \le O(\text{SDP-HC})$ holds from \hyperref[claim:better-SpCut]{Claim \ref{claim:better-SpCut}} (slightly modified).

\end{proof}

\begin{remark}
As in \hyperref[remark:cost-function]{Remark \ref*{remark:cost-function}}, here $f$ should be polynomially growing.
\end{remark}


\end{document}